\documentclass[a4paper,11pt]{article}
\usepackage{hyperref}
\usepackage[top=1 in,bottom=1 in,left=1.0 in,right=1.0 in]{geometry}
\hypersetup{hidelinks}
\usepackage{amsthm}
\usepackage{amssymb,amsmath}

\numberwithin{equation}{section}

\newtheorem{lemma}{Lemma}

\newtheorem{theorem}{Theorem}

\newcommand{\bs}{\boldsymbol}
\newcommand{\mb}{\mathbb}
\newcommand{\mf}{\mathbf}
\newcommand{\mr}{\mathrm}
\newcommand{\e}{\mathbf{e}}
\newcommand{\I}{\mathrm{i}}

\newcommand{\pa}{\partial}
\newcommand{\bsb}{\begin{subequations}}
\newcommand{\esb}{\end{subequations}}

\usepackage{amsmath}
\usepackage{amsfonts}
\usepackage{latexsym}
\usepackage{amssymb}
\usepackage{appendix}
\usepackage{mathrsfs}
\usepackage{graphicx}
\usepackage{bm}

\usepackage{color}

\begin{document}

\title{\textbf{Cauchy matrix approach to the SU(2) self-dual Yang--Mills equation}}

\author{Shangshuai Li$^1$,~~ Changzheng Qu$^2$,~~
Xiangxuan Yi$^3$,~~ Da-jun Zhang$^1$\footnote{Corresponding author. Email: djzhang@staff.shu.edu.cn} \\
{\small $~^1$Department of Mathematics, Shanghai University, Shanghai 200444,  China}\\
{\small $~^2$School of Mathematics and Statistics, Ningbo University, Ningbo 315211, China}\\
{\small $~^3$Qianweichang School, Shanghai University, Shanghai 200444,  China}}
\maketitle

\begin{abstract}
The Cauchy matrix approach is developed to solve the $\mf{SU}(2)$  self-dual Yang--Mills equation.
Starting from a Sylvester matrix equation
coupled with certain dispersion relation for infinite coordinates,
the self-dual Yang--Mills equation under Yang's formulation is constructed.
By imposing further constraints on complex independent variables,
a broad class of explicit solutions are obtained.

\begin{description}
\item[Keywords:] self-dual Yang-Mills equation, Cauchy matrix approach, solutions, integrable system
\item[PACS numbers:] 02.30.Ik, 02.30.Ks, 05.45.Yv
\end{description}
\end{abstract}

\section{Introduction}\label{sec-1}

The Yang--Mills theory is the most important development in physics in the second half of last century (see \cite{Wu-2003}).
The pioneer work is that of Yang and Mills \cite{YM-1954},
which laid a foundation  of nonabelian gauge theories that explain the electromagnetic, the strong and weak
nuclear interactions.
The theory has also geometric interpretations in its nature,
where the concept of gauge fields is found to be identical to fiber bundles,
e.g. \cite{At-1980,BL-1982,Gu-1981,WY-1975} (also see \cite{B-1987,We-2016} and the references therein).

The Yang--Mills equation of motion is \cite{P-1980} (please refer Section \ref{sec-2} for notations)
\begin{equation}\label{YM}
\partial_\mu F_{\mu\nu}+[B_\mu, F_{\mu\nu}]=[\mathcal{D}_\mu, F_{\mu\nu}]=0,
\end{equation}
which corresponds to the case where action functional $S$ and energy functional $E$
take local minima in the sense of semi-classical approximation.
It is hard to solve this equation exactly. Even one could find solutions to \eqref{YM},
it is still difficult to verify $S$ or $E$ reaches local minima.
However, it can be shown that (see \cite{BPST-1975,P-1980})
when the gauge field strength $F_{\mu\nu}$ is self dual, for given integer topological charge $q$,
the gauge fields are absolute minima of the action $S$.
The corresponding solutions are called instantons and monopoles (for static gauge fields).
The paper \cite{P-1980} well collected early approaches to the solutions
of the self-dual Yang--Mills (SDYM) equation before 1980,
such as the approach based on the so called Corrigan--Fairlie--'t Hooft--Wilczek ansatz \cite{CF-1977,Wilc-1977},
the Atiyah--Hitchin--Drinfeld--Manin construction \cite{AHDM-1978} on
instantons with integer topological charge $q$ and depending on (for
$\mathbf{SU}(\mathcal N)$ case) at least $8q-3$ parameters \cite{AHS-1977},
and the B\"acklund transformation approach \cite{CFYG-1978} based on the Atiyah--Ward ansatz \cite{AW-1977}.
Both  \cite{AHDM-1978} and \cite{AW-1977} follows Ward's observation in 1977 \cite{Ward-1977}
on the connection between self-dual gauge fields and twistor theory.
One can also refer \cite{P-1980} and the references therein for more details.

The SDYM equation is an integrable system. It has a Lax pair and has
Painlev\'e property for any gauge group \cite{JKM-1982,Ward-1982},
therefore some methods based on integrability
have been employed to solve the SDYM equation,
such as a direct transform approach by solving Lax pair \cite{BZ-1978},
B\"acklund transformation based on Riemann-Hilbert problem \cite{UN-1982},
an approach inspired from Sato's theory \cite{T-1984},
bilinear method  \cite{ONG-2001,SOM-1998},
and Darboux transformation \cite{NGO-2000}.

In this paper we aim to solve the $\mf{SU}(2)$ SDYM equation via a direct method,
the Cauchy matrix approach.
The Cauchy matrix approach is a method to construct and study integrable equations
by means of the Sylvester-type equations.
In this approach integrable equations are presented as closed forms of some recurrence relations
involving  derivatives (or shifts).
It was first systematically used in \cite{NAH-2009} to investigate integrable quadrilateral equations and later
developed in \cite{XZZ-2014,ZZ-2013} to more general cases.
 We will construct the $\mf{SU}(2)$ SDYM equation together with its explicit solutions.

The paper is organized as follows.
In Sec.\ref{sec-2} we recall Yang's formulation for the self-duality of the Yang-Mills equation.
Then, in Sec.\ref{sec-3} we make use of the Cauchy matrix approach
to construct a SDYM equation together with its solutions.
These solutions will be elaborated in Sec.\ref{sec-4} by imposing constraints so that they meet
self-duality of $\mf{SU}(2)$ gauge group.
Finally, concluding remarks are given in Sec.\ref{sec-5}.
There are three appendixes, which are devoted to proving Lemma \ref{lem-2-2},
presenting solutions of the Sylvester equation \eqref{Syl_eq},
and presenting examples of solutions of the SDYM equation.

\section{The SDYM equation: Yang's formulation}\label{sec-2}

Let $B_\mu$'s be matrix valued gauge potentials defined on $\mathbb{R}^4$
and $F_{\mu\nu}$ the gauge field strength defined by
\begin{align}\label{field-strength}
    F_{\mu\nu}\doteq\partial_\nu B_\mu-\partial_\mu B_\nu-[B_\mu,B_\nu],
\end{align}
where  $[\cdot,\cdot]$ is the Lie bracket defined as $[G,H]=GH-HG$,
and $\partial_\mu$ stands for the differential operator $\partial/\partial x^\mu$,
$(x^0, x^1, x^2, x^3)\in \mathbb{R}^4$.
For a  given $\mf{SU}(\mathcal N)$ gauge field defined on  $\mathbb{R}^4$,
the self-duality gives rise to
\begin{equation}\label{self-dual}
    F_{\mu\nu}=*F_{\mu\nu}\doteq\frac{1}{2}\epsilon_{\mu\nu\alpha\beta}F_{\alpha\beta},
\end{equation}
where $\epsilon_{\mu\nu\alpha\beta}$ is the Levi-Civita tensor,
$*F$ is the duality of field strength $F$, and $\mu,\nu,\alpha,\beta$ run over $\{0, 1,2,3\}$.
In this case, the gauge potentials $B_\mu$'s are also self-dual and belong to $\mf{su}(\mathcal N)$.
Yang extended $(x^0, x^1, x^2, x^3)$ from $\mathbb{R}^4$ to $\mathbb{C}^4$
and introduced coordinates (transformation) \cite{Yang-1977} (cf.\cite{Ward-1977})
\begin{align*}
    \bs Y=
    \sqrt2
    \begin{pmatrix}
        y & -\bar z \\
        z & \bar y
    \end{pmatrix}
    =x^0-\I\,\bs x\cdot\bs\sigma,
\end{align*}
i.e.
\begin{equation}
y=\frac{\sqrt2}{2}(x^0-\I x^3), ~ \bar y=\frac{\sqrt2}{2}(x^0+\I x^3),~
z=\frac{\sqrt2}{2}(x^2-\I x^1), ~ \bar z=\frac{\sqrt2}{2}(x^2+\I x^1),
\end{equation}
where $\I^2=-1$, $\bs x=(x^1,x^2,x^3)$, $\bs \sigma=(\sigma_1,\sigma_2,\sigma_3)$
and $\sigma_i$'s are Pauli matrices.
Note that here $\bar y, \bar z$ denote variables independent of the complex conjugates, $y^*, z^*$,
of $y$ and $z$, so that real Euclidean space is specified by $\bar y = y^*$ and $\bar z=z^*$.
The self-duality condition \eqref{self-dual} is then reduced to
\bsb\label{self-dual-expand-new}
\begin{align}
    \label{self-dual-expand-new1}F_{yz}=F_{\bar y\bar z}=0,\\
    \label{self-dual-expand-new2}F_{y\bar y}+F_{z\bar z}=0.
\end{align}
\esb
Next, introducing covariant derivative $\mathcal D_\chi:=\partial_\chi+B_\chi$,
$\chi\in\{y,\bar y,z,\bar z\}$ , where (cf.\cite{Po-1980})
\begin{align}
    B_y=B_{0}+\I B_{3},~~ B_{\bar y}=B_{0}-\I B_{3}, ~~
    B_z=B_{2}+\I B_{1}, ~~ B_{\bar z}=B_{2}-\I B_{1},
\end{align}
one obtains a   representation of \eqref{self-dual-expand-new1}
with respect to $\mathcal D_\chi$,
\begin{align}
    [\mathcal D_y,\mathcal D_z]=0, ~~~ [\mathcal D_{\bar y},\mathcal D_{\bar z}]=0.
\end{align}
This  implies that there exist two $\mathcal N\times \mathcal N$ generating matrices $D$ and $\bar D$ such that
\begin{align}\label{generating-matrices}
    \mathcal D_y(D)=\mathcal D_z(D)=0, ~~~
    \mathcal D_{\bar y}(\bar D)=\mathcal D_{\bar z}(\bar D)=0,
\end{align}
which leads to
\begin{align}\label{gauge-potential}
    B_y=D\pa_y(D^{-1}),~~B_z=D\pa_z(D^{-1}), ~~
     B_{\bar y}=\bar D\pa_{\bar y}(\bar D^{-1}), ~~ B_{\bar z}=\bar D\pa_{\bar z}(\bar D^{-1}).
\end{align}
Then, for the field strength $F_{\mu\nu}$ defined by \eqref{field-strength} with
the above $B_\chi$'s, equation \eqref{self-dual-expand-new1} is satisfied and \eqref{self-dual-expand-new2}
is cast into \cite{BFNY-1978}
\begin{align}\label{SDYM}
     (J_{\bar y}J^{-1})_y+(J_{\bar z}J^{-1})_z=0,
\end{align}
or its alternative form
\begin{align}\label{SDYM-a}
    (J^{-1}J_y)_{\bar y}+(J^{-1}J_z)_{\bar z}=0,
\end{align}
where
\begin{equation}
J=D\bar D^{-1}.
\end{equation}

When the gauge group is of $\mf{SU}(\mathcal N)$,
it turns out that \cite{Po-1980,Yang-1977} $D,\bar D\in\mf{SL}(\mathcal N)$ and $\bar D=(D^\dagger)^{-1}$,
where $D^\dagger=(D^*)^T$,
and hence $J=DD^\dagger$.
In other words, $J$ is a positive-definite Hermitian matrix with $|J|=1$.

Equation \eqref{SDYM} (or its alternative form) is usually the equation that was solved by using integrable methods,
see \cite{MZ-1981,NGO-2000,ONG-2001,SOM-1998}.
In $\mf{SU}(2)$ case, once $J$ is obtained, $D$ can be recovered by Yang's $R$-gauge \cite{Yang-1977} as the following.
Writing $J$ in the form
\begin{align}\label{J-expand}
    J=\frac{1}{f}
    \begin{pmatrix}
    1 & -g \\
    e & f^2-eg \\
    \end{pmatrix},
\end{align}
where $f$ is real and $e=-g^*$. Then $D$ takes the form
\begin{align}
    D=\frac{1}{\sqrt f}
    \begin{pmatrix}
        1 & 0\\
        e & f \\
    \end{pmatrix}U
\end{align}
where $U\in \mf{SU}(2)$.
Hence   $B_\chi$'s are recovered from \eqref{gauge-potential}
and so   are $F_{\mu\nu}$'s   from \eqref{field-strength}.

In this paper we will develop the Cauchy matrix approach to  construct equation \eqref{SDYM} together with
its explicit solution $J$.

\section{The Cauchy matrix approach to the SDYM equation}\label{sec-3}

In order to construct the $\mathbf{SU}(2)$ SDYM equation \eqref{SDYM},
we start from the Sylvester equation
\begin{align}\label{Syl_eq}
    \bs K\bs M-\bs M\bs K=\bs r\bs s^T,
\end{align}
dressed with dispersion relations
\begin{align}\label{rs_move}
    \bs r_{x_n}=\bs{AK}^n\bs r,~~ \bs s_{x_n}=\bs A(\bs K^T)^n\bs s, ~~(n\in\mb Z),
\end{align}
where $\{x_n\}$ are independent infinite complex variables,
$\bs K,\bs M,\bs A,\bs r,\bs s$ are block matrices in the form of
\begin{align}\label{KAM}
    \bs K=
    \begin{pmatrix}
        \bs K_1 & \bs 0 \\
        \bs 0 & \bs K_2 \\
    \end{pmatrix},
    ~
    \bs M=
    \begin{pmatrix}
        \bs 0 & \bs M_1 \\
        \bs M_2 & \bs 0 \\
    \end{pmatrix},
~
    \bs A=
    \begin{pmatrix}
        \bs I_{N_1} & \bs 0 \\
        \bs 0 & -\bs I_{N_2}
    \end{pmatrix},
~
    \bs r=
    \begin{pmatrix}
        \bs r_1 & \bs 0 \\
        \bs 0 & \bs r_2 \\
    \end{pmatrix},
 ~
    \bs s=
    \begin{pmatrix}
        \bs 0 & \bs s_1 \\
        \bs s_2 & \bs 0 \\
    \end{pmatrix},
\end{align}
with $\bs K_i\in \mb C_{N_i\times N_i},\bs M_1\in \mb C_{N_1\times N_2}[{\mathbf{x}}],
\bs M_2\in\mb C_{N_2\times N_1}[\mathbf{x}], \bs r_i,\bs s_i\in \mb C_{N_i\times 1}[\mathbf{x}]$,
$\bs I_{N_i}$ being the $N_i$-th order identity matrix for $i=1,2$,
$\mathbf{x}=(\cdots, x_{-1}, x_0, x_1,\cdots)$, and $N_1+N_2=2N$.
In addition, we  assume $\bs K_1$ and $\bs K_2$ are invertible and do not share any eigenvalues
so that   the Sylvester equation \eqref{Syl_eq} has a unique solution $\bs M$ for given
$\bs K,  \bs r,\bs s$  \cite{Syl}.
Define an infinite matrix $\bs S=(\bs S^{(i,j)})_{\infty\times  \infty}$ where
each $\bs S^{(i,j)}$ is a $2 \times 2$ matrix defined as
\begin{align}\label{Sij}
    \bs S^{(i,j)}\doteq \bs s^T\bs K^j(\bs I+\bs M)^{-1}\bs K^i\bs r=
    \begin{pmatrix}
        s_1^{(i,j)} & s_2^{(i,j)} \\
        s_3^{(i,j)} & s_4^{(i,j)} \\
    \end{pmatrix},
~~ (i,j\in\mb Z),
\end{align}
i.e.
\bsb\label{Sij_element}
\begin{align}
    s_1^{(i,j)}&=-\bs s_2^T \bs K_2^j(\bs I_{N_2}-\bs M_2\bs M_1)^{-1}\bs M_2\bs K_1^i\bs r_1, \\
    s_2^{(i,j)}&=\bs s_2^T\bs K_2^j(\bs I_{N_2}-\bs M_2\bs M_1)^{-1}\bs K_2^i\bs r_2, \\
    s_3^{(i,j)}&=\bs s^T_1\bs K_1^j(\bs I_{N_1}-\bs M_1\bs M_2)^{-1}\bs K_1^i\bs r_1, \\
    s_4^{(i,j)}&=-\bs s^T_1\bs K_1^j(\bs I_{N_1}-\bs M_1\bs M_2)^{-1}\bs M_1\bs K_2^i\bs r_2,
\end{align}
\esb
where $\bs I$  specially denotes the $2N$-th order identity matrix.
The above settings are the same as those for deriving the Ablowitz--Kaup--Newell--Segur (AKNS) system \cite{Zhao-2018},
except here in \eqref{rs_move} we have introduced variable $x_0$ and its dispersion relation.
Thus we may make use of the results already obtained in \cite{Zhao-2018}.
However,   the variable $x_0$ does play a useful role in our procedure
of deriving the SDYM equation \eqref{SDYM}.

For $\bs S^{(i,j)}$ defined in \eqref{Sij} where $\bs K, \bs M, \bs r, \bs s$ are governed by the Sylvester equation
\eqref{Syl_eq}, there exists a recursive relation independent of dispersion relations,
see Proposition 2 in \cite{Zhao-2018} and cf.\cite{XZZ-2014}.

\begin{lemma}\label{lem-2-1}
$\{\bs S^{(i,j)}\}$ defined in \eqref{Sij} satisfy
\begin{equation}
\bs S^{(i,j+s)}=\bs S^{(i+s,j)}-\sum^{s-1}_{l=0} \bs S^{(s-1-l,j)}\bs S^{(i,l)},~~ (s=1,2,\cdots),
\end{equation}
where $\bs K, \bs M, \bs r, \bs s$ obey the Sylvester equation \eqref{Syl_eq}.
In particular, when $s=1$, it reads
\begin{align}\label{Sij_re}
    \bs S^{(0,j)}\bs S^{(i,0)}=\bs S^{(i+1,j)}-\bs S^{(i,j+1)}.
\end{align}
\end{lemma}

With respect to the dispersion relation \eqref{rs_move}, $\bs S^{(i,j)}$ evolves as the following.
\begin{lemma}\label{lem-2-2}
$\bs S^{(i,j)}$ defined by \eqref{Sij} obeys evolutions
\begin{subequations}\label{Sij_moves}
\begin{align}
\bs S^{(i,j)}_{x_n}&=\bs S^{(i+n,j)}\bs a-\bs a\bs S^{(i,j+n)}
    -\sum_{l=0}^{n-1}\bs S^{(n-1-l,j)}\bs a\bs S^{(i,l)}, &(n\in\mb Z^+),\label{Sij_moves1} \\
\bs S^{(i,j)}_{x_0}&=\bs S^{(i,j)}\bs a-\bs a\bs S^{(i,j)}=[\bs S^{(i,j)},\bs a], &~\label{Sij_moves2}\\
\bs S^{(i,j)}_{x_n}&=\bs S^{(i+n,j)}\bs a-\bs a\bs S^{(i,j+n)}
    +\sum_{l=-1}^{n}\bs S^{(n-1-l,j)}\bs a\bs S^{(i,l)}, &(n\in\mb Z^-),\label{Sij_moves3}
\end{align}
\end{subequations}
where $\bs K, \bs M, \bs r, \bs s$ satisfy the Sylvester equation \eqref{Syl_eq}
and dispersion relation \eqref{rs_move},
and $\bs a=\sigma_3=\mf{diag}(1,-1)$.
\end{lemma}

Formulae \eqref{Sij_moves1} and \eqref{Sij_moves3} have been derived in \cite{Zhao-2018}
and \eqref{Sij_moves2} can be obtained similarly. For the completeness of the paper,
a proof of Lemma \ref{lem-2-2} is given in Appendix \ref{A}.

Now we come to the first main result of this paper.

\begin{theorem}\label{thm1}
Let
    \begin{align}\label{def_uv}
        \bs u\doteq \bs S^{(0,0)}=\bs s^T(\bs I+\bs M)^{-1}\bs r, ~~
        \bs v\doteq \bs I_2-\bs S^{(-1,0)}=\bs I_2- \bs s^T(\bs I+\bs M)^{-1}\bs K^{-1}\bs r.
    \end{align}
Then $\bs u$ and $\bs v$  satisfy the following differential recurrence relation
    \begin{align}\label{recur_uv}
        \bs v_{x_{n+1}}\bs v^{-1}=-\bs u_{x_n},~~~ (n\in\mb Z).
    \end{align}
\end{theorem}

\begin{proof}
First, the recursive formula \eqref{Sij_re} with $i=-1$ yields
\begin{equation}\label{S-v}
\bs S^{(-1,j+1)}=\bs S^{(0,j)}\bs v,
\end{equation}
which looks simple but will play a crucial role in the following proof.
It gives rise to (with $j=0$)
\begin{equation}\label{S-uv}
\bs u \bs v =\bs S^{(-1,1)}
\end{equation}
and (with $j=-1$)
\begin{equation}\label{S-v-1}
\bs v^{-1}=\bs I_2+\bs S^{(0,-1)}.
\end{equation}
In addition, the evolution relation \eqref{Sij_moves2} indicates (considering $(i,j)=(0,0)$ and $(-1,0)$)
\begin{equation}\label{u-x0}
\bs u_{x_{0}}=[\bs u, \bs a]
\end{equation}
and
\begin{equation}\label{v-x0}
\bs v_{x_{0}}=[\bs v, \bs a].
\end{equation}

Next, looking at \eqref{Sij_moves1} with $n=1$ and $(i,j)=(-1,0)$, and making use of
\eqref{S-uv} and \eqref{u-x0}, we have
\begin{align}
-\bs v_{x_1}=\bs u\bs a-\bs a\bs S^{(-1,1)}-\bs u\bs a(\bs I_2-\bs v)
=\bs u\bs a\bs v-\bs a\bs u\bs v=[\bs u,\bs a]\bs v=\bs u_{x_0}\bs v.
\end{align}
Similarly, with $n=-1$ and $(i,j)=(0,0)$ formula \eqref{Sij_moves3} gives rise to
\begin{align}
    \bs u_{x_{-1}}=\bs a-\bs v\bs a(\bs I_2+\bs S^{(0,-1)})=\bs a-\bs v\bs a\bs v^{-1}
    =(\bs a\bs v-\bs v\bs a)\bs v^{-1}=-\bs v_{x_0}\bs v^{-1},
\end{align}
where use has been made of \eqref{S-v-1} and \eqref{v-x0}.
These two equations cover the cases $n=0$ and $-1$ of \eqref{recur_uv}.

Next, we prove \eqref{recur_uv} for positive $n$. Formula \eqref{Sij_moves1}
gives rise to (with $i=j=0$)
\begin{align*}
        \bs u_{x_n}=\bs S^{(n,0)}\bs a-\bs a\bs S^{(0,n)}-\sum_{l=0}^{n-1}\bs S^{(n-1-l,0)}\bs a\bs S^{(0,l)}
    \end{align*}
and (with  $i=-1,j=0$)
\begin{align*}
-\bs v_{x_{n+1}}&=\bs S^{(n,0)}\bs a-\bs a\bs S^{(-1,n+1)}-\sum_{l=0}^n\bs S^{(n-l,0)}\bs a\bs S^{(-1,l)} \\
        &=\bs S^{(n,0)}\bs a-\bs a\bs S^{(0,n)}\bs v-\bs S^{(n,0)}\bs a\bs S^{(-1,0)}
        -\sum_{l=1}^n\bs S^{(n-l,0)}\bs a\bs S^{(-1,l)},
\end{align*}
where we have made use of \eqref{S-v} and separated the first term from the summation.
Then, in light of definition of $\bs v$, replacing the index $l$ with $l+1$
and using \eqref{S-v} once again, we arrive at
\begin{align*}
 -\bs v_{x_{n+1}}&=\bs S^{(n,0)}\bs a\bs v-\bs a\bs S^{(0,n)}\bs v
 -\sum_{l=0}^{n-1}\bs S^{(n-1-l,0)}\bs a\bs S^{(-1,l+1)} \\
        &=\bs S^{(n,0)}\bs a\bs v-\bs a\bs S^{(0,n)}\bs v-\sum_{l=0}^{n-1}\bs S^{(n-1-l,0)}\bs a\bs S^{(0,l)}\bs v
         =\bs u_{x_n}\bs v,
    \end{align*}
which is  \eqref{recur_uv} with $n\geq 1$.

The case of $n$ less than $-1$ can be proved similarly from \eqref{Sij_moves3}.
In details, we have
\[  \bs u_{x_n}=\bs S^{(n,0)}\bs a-\bs a\bs S^{(0,n)}+\sum_{l=-1}^n\bs S^{(n-1-l,0)}\bs a\bs S^{(0,l)}\]
and
\begin{align*}
-\bs v_{x_{n+1}}&=\bs S^{(n,0)}\bs a-\bs a\bs S^{(-1,n+1)}+\sum_{l=-1}^{n+1}\bs S^{(n-l,0)}\bs a\bs S^{(-1,l)} \\
        &=\bs S^{(n,0)}\bs a\bs v+\bs S^{(n,0)}\bs a\bs S^{(-1,0)}-\bs a\bs S^{(0,n)}\bs v
        +\sum_{l=-1}^{n+1}\bs S^{(n-l,0)}\bs a\bs S^{(-1,l)} \\
        &=\bs S^{(n,0)}\bs a\bs v-\bs a\bs S^{(0,n)}\bs v+\sum_{l=0}^{n+1}\bs S^{(n-l,0)}\bs a\bs S^{(-1,l)} \\
        &=\bs S^{(n,0)}\bs a\bs v-\bs a\bs S^{(0,n)}\bs v+\sum_{l=-1}^{n}\bs S^{(n-1-l,0)}\bs a\bs S^{(-1,l+1)} \\
        &=\bs S^{(n,0)}\bs a\bs v-\bs a\bs S^{(0,n)}\bs v+\sum_{l=-1}^{n}\bs S^{(n-1-l,0)}\bs a\bs S^{(0,l)}\bs v
        =\bs u_{x_n}\bs v.
\end{align*}

Thus, we have proved \eqref{recur_uv} for all $n\in \mathbb{Z}$.

\end{proof}

Considering the compatibility $(\bs u_{x_n})_{x_m}=(\bs u_{x_m})_{x_n}$,
we immediately arrive at the following.

\begin{theorem}\label{thm2}
For $\bs v\doteq \bs I_2-\bs S^{(-1,0)}=\bs I_2- \bs s^T(\bs I+\bs M)^{-1}\bs K^{-1}\bs r$
where $\bs K, \bs M, \bs r, \bs s$ satisfy the Sylvester equation \eqref{Syl_eq} and dispersion relation \eqref{rs_move},
the following relation holds,
\begin{align}\label{SDYM-3}
    (\bs v_{x_{n+1}}\bs v^{-1})_{x_m}-(\bs v_{x_{m+1}}\bs v^{-1})_{x_n}=0,
\end{align}
where $n,m\in \mathbb{Z}$.
\end{theorem}

As a by-product of \eqref{SDYM-3}, $\bs u$ satisfies a potential SDYM equation (cf.\cite{Lez-1987})
\begin{align}\label{PSDYM-3}
    \bs u_{x_{n},x_{m+1}}-\bs u_{x_{m},x_{n+1}}-[\bs u_{x_n},\bs u_{x_m}]=0.
\end{align}

Note that equation \eqref{SDYM-3} differs from the SDYM equation \eqref{SDYM} by a sign ``$-$''.
In next section, we will recover  \eqref{SDYM} from \eqref{SDYM-3} by imposing certain reductions.
Besides, apart from $\{x_n\}$ with  dispersion relation \eqref{rs_move}, we may introduce for $\{y_m\}$
such that
\begin{align}\label{rs_move-y}
    \bs r_{y_m}=\bs A(-\bs{K})^m\bs r,~~ \bs s_{y_m}=\bs A(-\bs K^T)^m\bs s, ~~(m\in\mb Z).
\end{align}
The resulting equation is
\begin{align}\label{SDYM-3y}
    (\bs v_{x_{n+1}}\bs v^{-1})_{y_m}+(\bs v_{y_{m+1}}\bs v^{-1})_{x_n}=0,
\end{align}
which is in a same form as \eqref{SDYM}.
In next section, we will work on \eqref{SDYM-3} and reduce it to \eqref{SDYM}.

\section{Solutions to the SU(2) SDYM equation}\label{sec-4}

When the gauge group is of $\mf{SU}(2)$,
$J$ in the SDYM equation \eqref{SDYM}
should be a positive-definite Hermitian matrix with $|J|=1$ \cite{Po-1980,Yang-1977}.
In the following we will look for a  Hermitian matrix $\bs v$ with $|\bs v|=1$.
This will be able to be achieved by imposing some constraints on $\bs K$ and $\{x_n\}$.

The Sylvester equation \eqref{Syl_eq} can be written as a more explicit form
\begin{subequations}\label{syl-2}
\begin{align}
& \bs K_1 \bs M_1- \bs M_1 \bs K_2=\bs r_1 \bs s^T_2, \label{syl-2a}\\
& \bs K_2 \bs M_2- \bs M_2 \bs K_1=\bs r_2 \bs s^T_1. \label{syl-2b}
\end{align}
\end{subequations}
For given invertible $\bs K_1$ and $\bs K_2$ that do not share any eigenvalues, $\bs M_1$ and $\bs M_2$
can be uniquely solved. In practice, since $\bs S^{(i,j)}$ is invariant with respect to $\bs K_1, \bs K_2$
and any matrices similar to them \cite{Zhao-2018} (cf.\cite{XZZ-2014,ZZ-2013}),
we can always consider $\bs K_1$ and $\bs K_2$ to be their canonical forms, say $\bs \Gamma$ and $\bs \Lambda$.
Solutions $\bs M_i$ together with $\bs r_i$ and $\bs s_i$ can be explicitly presented.
One may refer \cite{Zhao-2018} or Appendix \ref{B} of the present paper.
It turns out that these solutions can be presented via the following form:
\begin{align}\label{constructions}
    \bs M_1=\bs F_1\bs G_1\bs H_2, ~~ \bs M_2=\bs F_2\bs G_2\bs H_1, ~~ \bs r_1=\bs F_1\bs E_1,
    ~~\bs r_2=\bs F_2\bs E_2, ~~ \bs s_1=\bs H_1\bs E_1,~~ \bs s_2=\bs H_2\bs E_2,
\end{align}
and these elements satisfy (symmetric or commutative) relations
\begin{subequations}\label{GH}
\begin{align}
   & \bs G_1=-\bs G_2^T, ~~ \bs F_1\bs \Gamma=\bs \Gamma\bs F_1, ~~ \bs F_2\bs \Lambda=\bs \Lambda\bs F_2,
    ~~ \bs\Gamma\bs H_1=\bs H_1\bs\Gamma^T, ~~ \bs\Lambda\bs H_2=\bs H_2\bs\Lambda^T,
\\
   & \bs H_i^T=\bs H_i, ~~ (\bs H_i\bs F_i)^T=\bs F_i^T\bs H_i=\bs H_i\bs F_i, ~~ i=1,2,
\end{align}
\end{subequations}
where $\bs F_i$ and $\bs H_i$ are $N_i\times N_i$ matrices, $\bs G_1$ is a  $N_1\times N_2$ matrix,
$\bs G_2$ is a  $N_2\times N_1$ matrix, and $\bs E_i$ is a  $N_i$-th order column vector, $i=1,2$.

With the above notations, we are able to investigate symmetric property of the infinite matrix $\bs S$,
i.e. the relations between  $\bs S^{(i,j)}$ and  $\bs S^{(j,i)}$.

\begin{lemma}\label{lem-4-1}
$\{\bs S^{(i,j)}\}$ defined by \eqref{Sij}
with $\bs K_1=\bs \Gamma$, $\bs K_2=\bs \Lambda$, $\bs M, \bs r, \bs s$ satisfying the Sylvester equation \eqref{Syl_eq},
the elements in  $\bs S^{(i,j)}$ and  $\bs S^{(j,i)}$ are related as the following,
\begin{align}\label{sym}
    s_1^{(i,j)}=-s_4^{(j,i)}, ~~ s_2^{(i,j)}=s_2^{(j,i)}, ~~s_3^{(i,j)}=s_3^{(j,i)},~~ i,j \in \mathbb{Z},
\end{align}
i.e. ${\bs S^{(i,j)}}^T=-\sigma_2\,  \bs S^{(j,i)} \, \sigma_2$.
\end{lemma}

\begin{proof}
Making use of expressions \eqref{constructions} and relations \eqref{GH},
from \eqref{Sij_element} we find
\begin{align*}
            s_1^{(i,j)}&
            =-\bs s_2^T\bs \Lambda^j(\bs I_{N_2}-\bs M_2\bs M_1)^{-1}\bs M_2\bs \Gamma^i\bs r_1 \\
            &=(-\bs E^T_{2}(\bs \Lambda^T)^j((\bs H_2\bs F_2)^{-1}
            -\bs G_2\bs H_1\bs F_1\bs G_1)^{-1}\bs G_2\bs H_1\bs\Gamma^i\bs F_1\bs E_{1})^T \\
            &=\bs E_{1}^T\bs F_1^T(\bs\Gamma^T)^i\bs H_1\bs G_1((\bs H_2\bs F_2)^{-1}
            -\bs G_2\bs H_1\bs F_1\bs G_1)^{-1}\bs \Lambda^j\bs E_{2} \\
            &=\bs E_{1}^T\bs H_1\bs\Gamma^i\bs F_1\bs G_1\bs H_2
            (\bs I_{N_2}-\bs F_2\bs G_2\bs H_1\bs F_1\bs G_1\bs H_2)^{-1}
            \bs \Lambda^j\bs F_2\bs E_{2} \\
            &=\bs s_1^T\bs\Gamma^i\bs M_1(\bs I_{N_2}-\bs M_2\bs M_1)^{-1}\bs \Lambda^j\bs r_2
            =-s_4^{(j,i)},
\end{align*}
and
\begin{align*}
            s_2^{(i,j)}&=\bs s_2^T\bs \Lambda^j(\bs I_{N_2}-\bs M_2\bs M_1)^{-1}\bs \Lambda^i\bs r_2 \\
            &=\bs E^T_{2}\bs H_2\bs \Lambda^j(\bs I_{N_2}-\bs F_2\bs G_2\bs H_1\bs F_1\bs G_1\bs H_2)^{-1}
            \bs \Lambda^i\bs F_2\bs E_{2} \\
            &=(\bs E^T_{2}(\bs \Lambda^T)^j((\bs H_2\bs F_2)^{-1}-\bs G_2\bs H_1\bs F_1\bs G_1)^{-1}
            \bs \Lambda^i\bs E_{2})^T \\
            &=\bs E^T_{2}(\bs \Lambda^T)^i((\bs H_2\bs F_2)^{-1}-\bs G_2\bs H_1\bs F_1\bs G_1)^{-1}
            \bs \Lambda^j\bs E_{2} \\
            &=\bs s_2^T\bs \Lambda^i(\bs I_{N_2}-\bs M_2\bs M_1)^{-1}\bs \Lambda^j\bs r_2
            =s_2^{(j,i)}.
        \end{align*}
The third relation $s_3^{(i,j)}=s_3^{(j,i)}$ can be proved in a similar way.

\end{proof}

With this lemma we are able to prove $|\bs v|=1$.
In fact, the relation \eqref{S-v-1} yields
\begin{align*}
    (1+s_1^{(0,-1)})(1-s_1^{(-1,0)})-s_2^{(0,-1)}s_3^{(-1,0)}=1.
\end{align*}
Then, by Lemma \ref{lem-4-1} we may replace $s_1^{(0,-1)}$ and $s_2^{(0,-1)}$
using \eqref{sym} and the resulting equation gives rise to
\begin{align*}
\begin{vmatrix}
        1-s_1^{(-1,0)} & -s_2^{(-1,0)}\\
        -s_3^{(-1,0)} & 1-s_4^{(-1,0)}
    \end{vmatrix}
    =|\bs I_2-\bs S^{(-1,0)}|=|\bs v|=1.
\end{align*}

Next, we make $\bs v$ to be a Hermitian matrix by imposing constraints.
First, we introduce
\begin{align}\label{coor}
z_n\doteq x_n=\xi_n+\I \eta_n, ~~ \bar{z}_n \doteq  (-1)^{n+1}x_{-n}=\xi_n-\I \eta_n,
~~ n=1,2,\cdots
\end{align}
where  $ \xi_n, \eta_n \in\mb R$.
This indicates $\bar{z}_n=z^*_n$.
Then we take $m=-n-1$. The resulting equation \eqref{SDYM-3} reads
\begin{align}\label{SDYM-4}
    (\bs v_{z_{n+1}}\bs v^{-1})_{\bar{z}_{n+1}}+(\bs v_{\bar{z}_n}\bs v^{-1})_{z_n}=0, ~~ n=1,2,\cdots,
\end{align}
which coincides with the form \eqref{SDYM}.
We next introduce further constraints by $N_2=N_1$ and
\begin{equation}
\bs K_2=-(\bs K_1^*)^{-1}
\end{equation}
to the Cauchy matrix scheme (\ref{Syl_eq},\ref{rs_move},\ref{KAM}).
Equation \eqref{coor} implies
\begin{equation}
\partial_{\xi_n}=\partial_{z_n}+\partial_{\bar{z}_n},~~
\partial_{\eta_n}=\mathrm{i} (\partial_{z_n}-\partial_{\bar{z}_n}),
\end{equation}
and the dispersion relation \eqref{rs_move} is equivalently written in terms of $\xi_n$ and $\eta_n$ as
\begin{subequations}
\begin{align}
& \partial_{\xi_n} \bs r_1=(\bs K_1^n+(-1)^{n+1}\bs K_1^{-n})\bs r_1, ~~~
\partial_{\eta_n} \bs r_1= \mathrm{ i}(\bs K_1^n-(-1)^{n+1}\bs K_1^{-n})\bs r_1,\\
&  \partial_{\xi_n} \bs r_2=((\bs K_1^*)^n+(-1)^{n+1}(\bs K_1^*)^{-n})\bs r_2, ~~~
\partial_{\eta_n} \bs r_2= -\mathrm{ i}((\bs K_1^*)^n-(-1)^{n+1}(\bs K_1^*)^{-n})\bs r_2,
\end{align}
\end{subequations}
and
\begin{subequations}
\begin{align}
& \partial_{\xi_n} \bs s_1=((\bs K_1^T)^n+(-1)^{n+1}(\bs K_1^T)^{-n})\bs s_1, ~~~
\partial_{\eta_n} \bs s_1= \mathrm{ i}((\bs K_1^T)^n-(-1)^{n+1}(\bs K_1^T)^{-n})\bs s_1,\\
&  \partial_{\xi_n} \bs s_2=((\bs K_1^\dagger)^n+(-1)^{n+1}(\bs K_1^\dagger)^{-n})\bs s_2, ~~~
\partial_{\eta_n} \bs s_2= -\mathrm{ i}((\bs K_1^\dagger)^n-(-1)^{n+1}(\bs K_1^\dagger)^{-n})\bs s_2.
\end{align}
\end{subequations}
We are able to take\footnote{This means the reduction \eqref{coor} is allowed.}

\[\bs r_2=(\bs K_1^*)^{-1}\bs r_1^*,~~ \bs s_2=-\delta (\bs K_1^\dagger)^{-1} \bs s_1^*, \]
where $\delta=\pm 1$, such that
\[\bs M_2=\delta \bs M_1^*\]
in light of the uniqueness of solutions of the Sylvester equations \eqref{syl-2}.
Denote $\bs v= \Bigl(\begin{smallmatrix} v_1& v_2\\v_3 &  v_4 \end{smallmatrix}\Bigr)$.
Then, by direct calculation we find
\begin{align*}
    v_1^*&=1-(s_1^{(-1,0)})^*
    =1+(s_2^T)^*\bs M^*_2(\bs I_{N_2}-\bs M_1^*\bs M_2^*)^{-1}(\bs K_1^*)^{-1}\bs r^*_1 \\
    &=1-\bs s_1^T\bs K_1^{-1}\bs M_1(\bs I_{N_2}-\bs M_2\bs M_1)^{-1}\bs r_2
    =1+s_4^{(0,-1)}
    =1-s_1^{(-1,0)}=v_1,
\end{align*}
and $v^*_4=v_4$  in a similar way way, and
\begin{align*}
    v_2^*&=-(s_2^{(-1,0)})^*
    =-(\bs s_2^T)^*(\bs I_{N_2}-\bs M_2^*\bs M_1^*)^{-1}(\bs K_2^*)^{-1}\bs r^*_2 \\
    &=-\delta\bs s_1^T\bs K_1^{-1}(\bs I_{N_2}-\bs M_1\bs M_2)^{-1}\bs r_1
    =-\delta s_3^{(0,-1)}=-\delta s_3^{(-1,0)}=\delta v_3.
\end{align*}
All these together indicate $\bs v=\bs v^{\dagger}$ when we take  $\delta=1$,
i.e. $\bs v$ is a Hermitian matrix when $\delta=1$.

We end up the section with the following summarization.

\begin{theorem}\label{thm3}
The SDYM equation \eqref{SDYM-4} has the following solutions
\begin{equation}\label{v-su2}
\bs v= \bs I_2-\bs S^{(-1,0)}=
\begin{pmatrix}
    1-s_1^{(-1,0)} & -(s_3^{(-1,0)})^* \\
    -s_3^{(-1,0)} & 1-s_4^{(-1,0)}
\end{pmatrix},
\end{equation}
where
\begin{align}\label{coor2}
z_n =\xi_n+\I \eta_n, ~~ \bar{z}_n =z^*_n=\xi_n-\I \eta_n,
~~ n=1,2,\cdots,
\end{align}
$ \xi_n$ and $\eta_n$ are real, and
\bsb\label{s-v}
\begin{align}
    s_1^{(-1,0)}&=\bs s_1^\dagger (\bs K_1^*)^{-1}(\bs I_N-\bs M_1^*\bs M_1)^{-1}\bs M_1^*\bs K_1^{-1}\bs r_1, \\
    s_3^{(-1,0)}&=\bs s^T_1(\bs I_N-\bs M_1\bs M_1^*)^{-1}\bs K_1^{-1}\bs r_1, \\
    s_4^{(-1,0)}&=\bs s^T_1\bs M_1(\bs I_N-\bs M_1^*\bs M_1)^{-1}\bs r_1^*.
\end{align}
\esb
Here, $\bs K_1\in \mathbb{C}_{N\times N}$, $\bs K_1$ and $-(\bs K_1^*)^{-1}$ do not share any eigenvalues,
$\bs M_1, \bs r_1$ and $\bs s_1$ are determined by the system
\begin{subequations}\label{Syl-Th}
\begin{align}
& \bs K_1 \bs M_1+ \bs M_1 (\bs K_1^*)^{-1}=-\bs r_1 \bs s_1^\dagger (\bs K_1^*)^{-1}, \label{syl-tha}\\
& \partial_{\xi_n} \bs r_1=(\bs K_1^n+(-1)^{n+1}\bs K_1^{-n})\bs r_1, ~~~
\partial_{\eta_n} \bs r_1= \mathrm{ i}(\bs K_1^n-(-1)^{n+1}\bs K_1^{-n})\bs r_1,\\
& \partial_{\xi_n} \bs s_1=((\bs K_1^T)^n+(-1)^{n+1}(\bs K_1^T)^{-n})\bs s_1, ~~~
\partial_{\eta_n} \bs s_1=\mathrm{ i}((\bs K_1^T)^n-(-1)^{n+1}(\bs K_1^T)^{-n})\bs s_1,
\end{align}
\end{subequations}
for $n=1,2,\cdots$. Solution $\bs v$ satisfies $\bs v=\bs v^\dagger$ and $|\bs v|=1$.
$\bs v$ is piecewisely positive-definite or negative-definite, depending on the domains where
$\bs v$ is positive or negative.
\end{theorem}

Explicit solutions of the system \eqref{Syl-Th} can be formulated from Appendix \ref{B}.
Some examples of solution $\bs v$ will be listed in Appendix \ref{C}.

\section{Concluding remarks}\label{sec-5}

In this paper we have constructed explicit solution $\bs v$ for the SDYM equation \eqref{SDYM-4}
by using a direct method, namely, the Cauchy matrix approach.
We started with the Sylvester equation \eqref{Syl_eq} together with the dispersion relation \eqref{rs_move}
with respect to infinite coordinates $\{x_n\}$,
made use of the recursive relation  \eqref{Sij_re},
proved $\bs u$ and $\bs v$ satisfy the key equation \eqref{recur_uv}
that gives rise to the equation \eqref{SDYM-3}.
Then, we introduced real independent coordinates $\xi_n$ and $\eta_n$,
and imposed constraints on $z_n$, $\bar{z}_n$ and $\bs K_1$ and $\bs K_2$.
Finally, we obtained solution $\bs v$ that solves the SDYM equation \eqref{SDYM-4}.
Solutions have been presented via Theorem \ref{thm3}.

Our Cauchy matrix approach is based on the scheme for the AKNS system, cf.\cite{Zhao-2018}.
Compared with \cite{Zhao-2018}, here we introduced the auxiliary variable $x_0$,
frequently made use of the recursive relation  \eqref{Sij_re},
and finally we were able to construct the key equation \eqref{recur_uv}.
We also discussed symmetric relation of $\bs S^{(i,j)}$ and $\bs S^{(j,i)}$.
All these elaborations enabled us to finally prove the property $\bs v=\bs v^\dagger$ and $|\bs v|=1$, and
obtain exact solutions to the SDYM equation \eqref{SDYM-4}.

There are similar direct approaches to construct equation \eqref{recur_uv},
e.g., the one based on bidifferential graded algebra \cite{DMH-2008,DMH-2009},
where it is assumed there exists a function $\bs v$ to satisfy the equation \eqref{recur_uv}.
In our approach, $\bs v$ is clearly defined by $\bs I_2-\bs S^{(-1,0)}$,
which provides an explicit solution to the SDYM equation \eqref{SDYM-4}

In our scheme, after introducing $z_n$ and  $\bar{z}_n$ by \eqref{coor}, we took $m=-n-1$.
Instead of doing that, if we take $m=n-1$ in equation \eqref{SDYM-3}, we have an equation
\begin{align}\label{SDYM-3D}
    (\bs v_{x_{n+1}}\bs v^{-1})_{x_{n-1}}-(\bs v_{x_n}\bs v^{-1})_{x_n}=0, ~~~ (n\in\mb Z),
\end{align}
and by redefining $x_n$ by $ix_n$, it becomes \cite{MZ-1981}
\begin{align}\label{SDYM-3D-MZ}
    (\bs v_{x_{n+1}}\bs v^{-1})_{x_{n-1}}+(\bs v_{x_n}\bs v^{-1})_{x_n}=0, ~~~ (n\in\mb Z).
\end{align}
Manakov and Zakharov constructed its solutions using its Lax pair \cite{MZ-1981}.
Since the equation depends only on 3 independent variables, its solutions may generate monopoles (see \cite{P-1980}).
Solutions obtained in \cite{MZ-1981} are different from ours.
In fact, the two constraints on $n$ and $m$ hold simultaneously  only when $n=0$.
In this case, our results can be applied and the equation \eqref{SDYM-3D-MZ} will be written as
\begin{align}\label{SDYM-2D}
    (\bs v_{x_1}\bs v^{-1})_{x_{-1}}-[[\bs v,\bs a]\bs v^{-1},\bs a]=0,
\end{align}
where we have made use of relation \eqref{v-x0}.
Note that this equation is a 2D equation and was recently derived in \cite{Vek-2020}
by means of a similar direct method.

With regard to the methods of solving the SDYM equation, most of them are direct and constructive.
The approach in \cite{BZ-1978} can be thought of a Darboux-B\"acklund transformation
employing the Lax pair of the SDYM equation. Solutions obtained are algebraic type, i.e.
rational solutions in terms of polynomials of independent of variables.
The bilinear approach given in \cite{SOM-1998} (cf.\cite{ONG-2001}) is to reformulate the bilinear relations satisfied
by the functions $e, f, g$ in \eqref{J-expand} to a new set of bilinear equations involving nine $\tau$ functions
in Hankelians. In principle, they provide algebraic solutions as well.
Our solutions are expressed in terms of exponential functions,
which are in expression formally similar to those obtained via Darboux transformation \cite{NGO-2000}.
However, our solutions are different from those in \cite{NGO-2000},
as $J$ is a solution of the  $\mf{SU}(2)$ SDYM equation
but it is not necessary $J\in \mf{SU}(2)$, cf.\cite{NGO-2000}.
In addition, our expression for $\bs v$ is more explicit and includes
solutions generated by possible canonical forms of $\bs K_1$
(not only diagonal form or a Jordan form, but also any combinations of them).

There are several further investigations related to approach and results of the present paper.
For example, extend the approach to the $\mf{SU}(\mathcal N)$ SDYM equation
and noncommutative case, e.g.\cite{GHN-2009,GHHN-2020,T-2001}.
Besides, it is well known that as a 4D integrable system, the SDYM equation allows
various reductions to lower dimensional integrable equations \cite{ACT-1993,Cha-1993,Mas-1989,Mas-1993,Ward-1985}.
That would be interesting to understand how reductions play roles in generating solutions to
lower dimensional integrable equations.
In addition, Ward used to discuss discretisation of the SDYM equation \cite{Ward-1993}.
Since the Cauchy matrix approach originated from solving discrete integrable systems \cite{NAH-2009},
it would be interesting to have a discrete analogue of the SDYM equation from this approach.

\vskip 20pt
\subsection*{Acknowledgements}
This project is supported by the NSF of China (Nos. 11631007, 11875040, 11971251)
and Science and technology innovation plan of Shanghai (No. 20590742900).

\vskip 20pt
\begin{appendices}

\section{Proof for Lemma \ref{lem-2-2}}\label{A}

First, the following recurrence relations hold (see equation (2.1) in \cite{Zhao-2018}, cf.\cite{XZZ-2014}),
\bsb\label{recur_re}
\begin{align}
    \label{recur_re1}&\bs K^n\bs M-\bs M\bs K^n=\sum_{l=0}^{n-1}\bs K^{n-1-l}\bs r\bs s^T\bs K^l,
    ~~~  (n \in \mb Z^+), \\
    \label{recur_re2}&\bs K^n\bs M-\bs M\bs K^n=-\sum_{l=-1}^{n}\bs K^{n-1-l}\bs r\bs s^T\bs K^l,
    ~~~  (n \in \mb Z^-).
\end{align}
\esb
In addition, in light of the dispersion relation \eqref{rs_move}, we have
\begin{align*}
        \bs K\bs M_{x_n}-\bs M_{x_n}\bs K&=\bs r_{x_n}\bs s^T+\bs r\bs s^T_{x_n} \\
        &=\bs A\bs K^n\bs r\bs s^t+\bs r\bs s^T\bs K^n\bs A \\
        &=\bs A\bs K^n(\bs K\bs M-\bs M\bs K)-\bs A(\bs K\bs M-\bs M\bs K)\bs K^n \\
        &=\bs K\bs A(\bs K^n\bs M-\bs M\bs K^n)-\bs A(\bs K^n\bs M-\bs M\bs K^n)\bs K.
    \end{align*}
Note that we have  assumed $\bs K_1$ and $\bs K_2$ are invertible and do not share any eigenvalues
so that  the Sylvester equation \eqref{Syl_eq} has a unique solution $\bs M$ for given
$\bs K,  \bs r,\bs s$.
With such a property, it follows that
\[\bs M_{x_n}=\bs A(\bs K^n\bs M-\bs M\bs K^n),~~ n\in \mathbb{Z},\]
which, together with \eqref{recur_re}, gives rise to
\bsb\label{M_evo}
\begin{align}
    \label{M_evo1}&\bs M_{x_n}=\sum^{n-1}_{l=0}\bs K^{n-1-l}\bs r\bs a\bs s^T\bs K^{l}, &  (n&\in \mb Z^+), \\
    \label{M_evo2}&\bs M_{x_n}=\bs 0, &  (n&=0), \\
    \label{M_evo3}&\bs M_{x_n}=-\sum^{n}_{l=-1}\bs K^{n-1-l}\bs r\bs a\bs s^T\bs K^{l}. & (n&\in \mb Z^-),
\end{align}
\esb
Next, direct calculation yields
\begin{align*}
        \bs S^{(i,j)}_{x_n}&=\bs s^T_{x_n}\bs K^j(\bs I+\bs M)^{-1}\bs K^i\bs r
        +\bs s^T\bs K^j(\bs I+\bs M)^{-1}\bs K^i\bs r_{x_n}
        +\bs s^T\bs K^j((\bs I+\bs M)^{-1})_{x_n}\bs K^i\bs r  \\
        &=-\bs a\bs S^{(i,j+n)}+\bs S^{(i+n,j)}\bs a-\bs s^T\bs K^j(\bs I+\bs M)^{-1}
        \bs M_{x_n}(\bs I+\bs M)^{-1}\bs K^i\bs r,
\end{align*}
which gives rise to equations \eqref{Sij_moves} after substituting \eqref{M_evo} into it.

\section{Solutions to \eqref{syl-2} and notations}\label{B}

When $\bs K_1$ and $\bs K_2$ in the Sylvester equations \eqref{syl-2} take their canonical forms $\bs\Gamma$
and $\bs\Lambda$, solutions to \eqref{syl-2}  can be represented as \eqref{constructions}.
The involved notations are the following.
Let
\begin{align}
    \bs\Gamma=\mf{diag}(\bs\Gamma_{n_1}(k_1),\bs\Gamma_{n_2}(k_2),\cdots,\bs\Gamma_{n_p}(k_p)), ~~
    \bs\Lambda=\mf{diag}(\bs\Gamma_{m_1}(l_1),\bs\Gamma_{m_2}(l_2),\cdots,\bs\Gamma_{m_q}(l_q)),
\end{align}
where $\bs\Gamma_n(k)$ denotes a  $n$-th order Jordan block
\begin{align*}
    \bs\Gamma_{n}(k)=\begin{pmatrix}
        k & 0  & 0  & \cdots & 0 & 0 \\
        1   & k & 0 & \cdots & 0 & 0   \\
        0   &   1   &  k & \cdots & 0 & 0  \\
        \vdots    & \vdots  & \vdots  & \vdots & \vdots  & \vdots   \\
        0   & 0 & 0 & \cdots & 1 & k
    \end{pmatrix}_{n \times n},
\end{align*}
the index $\{n_i,m_j\}$ are positive integers\footnote{Here $n_i$ and $m_j$ allow to be 1,
which provides diagonal matrix blocks in $\bs\Gamma$ and $\bs\Lambda$.
In particular, when all $\{n_i\}$ are one, $\bs\Gamma=\mf{diag}(k_1, k_2, \cdots, k_N)$.
 }
and satisfy $\sum_{i=1}^pn_i=N_1$, $\sum_{j=1}^qm_j=N_2$.
Introduce a lower triangular Toeplitz matrix
\begin{align*}
\bs F_M(\rho(k))=
        \begin{pmatrix}
            \rho(k) & 0 & 0 & \cdots & 0 \\
            \frac{\partial_{l}\rho}{1!} & \rho(k) & 0 & \cdots & 0 \\
            \frac{\partial^2_{l}\rho}{2!} & \frac{\partial_{l}\rho}{1!} & \rho(k) & \cdots & 0 \\
            \vdots & \vdots & \vdots & \ddots & \vdots \\
            \frac{\partial^{M-1}_{l}\rho}{(M-1)!} & \frac{\partial^{M-2}_{l}\rho}{(M-2)!}
            & \frac{\partial^{M-3}_{l}\rho}{(M-3)!}  & \cdots & \rho(k)
        \end{pmatrix},
    \end{align*}
and a symmetric matrix
\begin{align*}
    \bs H_{M'}(\sigma(l))=\begin{pmatrix}
        \sigma(l) & \frac{\partial_{l}\sigma}{1!} & \frac{\partial^2_{l}\sigma}{2!} & \dots
        & \frac{\partial^{M'-1}_{l}\sigma}{(M'-1)!} \\
    \frac{\partial_{l}\sigma}{1!} & \frac{\partial^2_{l}\sigma}{2!} & \frac{\partial^3_{l}\sigma}{3!} & \dots & 0 \\
    \frac{\partial^2_{l}\sigma}{2!} & \frac{\partial^3_{l}\sigma}{3!} & \frac{\partial^4_{l}\sigma}{4!} & \dots & 0
        \\
        \vdots & \vdots & \vdots & \ddots & \vdots \\
        \frac{\partial^{M'-1}_{l}\sigma}{(M'-1)!} & 0 & 0 & \dots & 0 \\
    \end{pmatrix},
\end{align*}
where the plane wave factors are given by
\begin{align}
    \rho(k_i)=\exp\left(\sum_{n\in\mb Z}k_i^nx_n\right)\rho_i^{(0)}, ~~
    \sigma(l_j)=\exp\left(-\sum_{n\in\mb Z}l_j^nx_n\right)\sigma_j^{(0)},~~~
    \rho_i^{(0)}, \sigma_i^{(0)}\in \mathbb{C}.
\end{align}
Let
\[ \e_M=(\underbrace{1,0,0,\cdots,0}_{\text{M-dimensional}})^T,~~~
\{\bs G_{n,m}(k,l)\}_{ij}=
\tbinom{i-1}{i+j-2}
\frac{(-1)^{i+j}}{(k-l)^{i+j-1}}.\]
Then the basic elements in \eqref{constructions} are expressed as
\begin{align*}
    &\bs F_1=\mf{diag}(\bs F_{n_1}(\rho(k_1)),\dots,\bs F_{n_p}(\rho(k_p))),
    & &\bs H_1=\mf{diag}(\bs H_{n_1}(\rho(k_1)),\dots,\bs H_{n_p}(\rho(k_p))), \\
    &\bs F_2=\mf{diag}(\bs F_{m_1}(\sigma(l_1)),\dots,\bs F_{m_q}(\sigma(l_q))),
    & &\bs H_2=\mf{diag}(\bs H_{m_1}(\sigma(l_1)),\dots,\bs H_{m_q}(\sigma(l_q))), \\
    &\bs E_1^T=(\e_{n_1}^T,\dots,\e_{n_p}^T), & & \bs E_2^T=(\e_{m_1}^T,\dots,\e_{m_q}^T), \\
    &(\bs G_1)_{i,j}=\bs G_{n_i,m_j}(k_i,l_j), & &(\bs G_2)_{j,i}=\bs G_{m_j,n_i}(l_j,k_i).
\end{align*}

\section{Examples of solutions}\label{C}

\subsection{One-soliton solution}
When $N=1$  we have
\begin{align*}
    \bs K_1=k,~~ \bs r_1=\rho,~~ \bs s_1=\sigma,~~ \bs M_1=-(\rho\sigma^*)/(|k|^2+1),
\end{align*}
where (noting that the only difference between $\rho$ and $\sigma$ is the phase factor $\rho^{(0)}(k)$
and $\sigma^{(0)}(k)$)
\begin{subequations}\label{ps}
\begin{align}
    &\rho=\rho(k)=\exp\big((k^n+(-1)^{n+1}k^{-n})\xi_n+\mathrm{i}(k^n-(-1)^{n+1}k^{-n})\eta_n\big)
    \rho^{(0)}(k), \\
    &\sigma=\sigma(k)=\exp\big((k^n+(-1)^{n+1}k^{-n})\xi_n+\mathrm{i}(k^n-(-1)^{n+1}k^{-n})\eta_n\big)
    \sigma^{(0)}(k).
\end{align}
\end{subequations}
Hence the 1-soliton solution is given by
\begin{align*}
    \bs v=
\begin{pmatrix}
    1-s_1^{(-1,0)} & -(s_3^{(-1,0)})^* \\
    -s_3^{(-1,0)} & 1-s_4^{(-1,0)}
\end{pmatrix},
\end{align*}
where
\begin{align*}
    &s_1^{(-1,0)}=-\frac{1}{|k|^2}\frac{|\rho\sigma|^2(|k|^2+1)}{(|k|^2+1)^2-|\rho\sigma|^2},\\ &s_3^{(-1,0)}=\frac{1}{k}\frac{\rho\sigma(|k|^2+1)^2}{(|k|^2+1)^2-|\rho\sigma|^2},\\
    &s_4^{(-1,0)}=-\frac{|\rho\sigma|^2(|k|^2+1)}{(|k|^2+1)^2-|\rho\sigma|^2}.
\end{align*}

\subsection{Two-soliton solution}
In the case $N=2$, we assume
\begin{align*}
    \bs K_1=\begin{pmatrix}
        k & 0\\
        0 & l
    \end{pmatrix},~~
    \bs r_1=\begin{pmatrix}
        \rho_{k}  \\
        \rho_{l}
    \end{pmatrix},~~
    \bs s_1=\begin{pmatrix}
        \sigma_{k}  \\
        \sigma_{l}
    \end{pmatrix},
\end{align*}
and we have
\begin{align*}
    \bs M_1=\begin{pmatrix}
        m_{11} & m_{12} \\
        m_{21} & m_{22}
    \end{pmatrix}
    =\begin{pmatrix}
        -(\rho_{k}\sigma_{k}^*)/(|k|^2+1) & -(\rho_{k}\sigma_{l}^*/(kl^*+1) \\
        -(\rho_{l}\sigma_{k}^*)/(k^*l+1) & -(\rho_{l}\sigma_{l}^*/(|l|^2+1)
    \end{pmatrix}, 
\end{align*}
where
\begin{align*}
    &\rho_{k}= \exp\big((k^n+(-1)^{n+1}k^{-n})\xi_n+\mathrm{i}(k^n-(-1)^{n+1}k^{-n})\eta_n\big)
    \rho^{(0)}_k, \\
    &\sigma_{k}=\exp\big((k^n+(-1)^{n+1}k^{-n})\xi_n+\mathrm{i}(k^n-(-1)^{n+1}k^{-n})\eta_n\big)
    \sigma^{(0)}_k.
\end{align*}
Introduce
\begin{align*}
    \bs T=(\bs I_2-\bs M_1^*\bs M_1)^{-1}=\frac{1}{\tau}\begin{pmatrix}
        T_{11} & T_{12} \\
        T_{21} & T_{22}
    \end{pmatrix},
\end{align*}
where
\begin{align*}
    \tau=|\bs I_2-\bs M_1^*\bs M_1|=1&-\frac{\rho_k\sigma_k\rho_l^*\sigma_l^*}{(kl^*+1)^2}-\frac{|\rho_l\sigma_l|^2}{(|l|^2+1)^2}
    -\frac{|\rho_k\sigma_k|^2}{(|k|^2+1)^2}-\frac{\rho_k^*\sigma_k^*\rho_l\sigma_l}{(k^*l+1)^2} \\
    &+\frac{[(|k|^2+1)(|l|^2+1)-(k^*l+1)(kl^*+1)]^2}{(|l|^2+1)^2(|k|^2+1)^2(kl^*+1)^2(k^*l+1)^2},
\end{align*}
and
\begin{align*}
    T_{11}=1-\frac{\rho_k\sigma_k\rho_l^*\sigma_l^*}{(kl^*+1)^2}-\frac{|\rho_l\sigma_l|^2}{(|l|^2+1)^2}, ~~
    T_{12}=\frac{\rho_k\rho_l^*|\sigma_k|^2}{(|k|^2+1)(kl^*+1)}
    +\frac{|\rho_l|^2\sigma_k^*\sigma_l}{(|l|^2+1)(k^*l+1)}, \\
    T_{21}=\frac{|\rho_k|^2\sigma_k\sigma_l^*}{(|k|^2+1)(kl^*+1)}
    +\frac{\rho_k^*\rho_l|\sigma_l|^2}{(k^*l+1)(|l|^2+1)}, ~~
    T_{22}=1-\frac{|\rho_k\sigma_k|^2}{(|k|^2+1)^2}-\frac{\rho_k^*\sigma_k^*\rho_l\sigma_l}{(k^*l+1)^2}.
\end{align*}
In addition, let
\bsb\label{PQs}
\begin{align}
    \bs P=(\bs I_N-\bs M_1^*\bs M_1)^{-1}\bs M_1^*=\frac{1}{\tau}\begin{pmatrix}
        P_{11} & P_{12} \\
        P_{21} & P_{22}
    \end{pmatrix}=
    \frac{1}{\tau}\begin{pmatrix}
        T_{11}m_{11}^*+T_{12}m_{21}^* & T_{11}m_{12}^*+T_{12}m_{22}^* \\
        T_{21}m_{11}^*+T_{22}m_{21}^* & T_{21}m_{12}^*+T_{22}m_{22}^*
    \end{pmatrix},\\
    \bs Q=\bs M_1(\bs I_N-\bs M_1^*\bs M_1)^{-1}=\frac{1}{\tau}\begin{pmatrix}
        Q_{11} & Q_{12} \\
        Q_{21} & Q_{22}
    \end{pmatrix}=
    \frac{1}{\tau}\begin{pmatrix}
        m_{11}T_{11}+m_{12}T_{21} & m_{11}T_{12}+m_{12}T_{22} \\
        m_{21}T_{11}+m_{22}T_{21} & m_{21}T_{12}+m_{22}T_{22}
    \end{pmatrix}.
\end{align}
\esb
Those elements in $\bs v$ turn out to be
\bsb\label{s134}
\begin{align}
    & s_1^{(-1,0)}=\frac{1}{\tau}(\frac{1}{|k|^2}\sigma_k^*P_{11}\rho_k
    +\frac{1}{kl^*}\sigma_l^*P_{21}\rho_k+\frac{1}{k^*l}\sigma_k^*P_{12}\rho_l
    +\frac{1}{|l|^2}\sigma_l^*P_{22}\rho_l),\\
    &s_3^{(-1,0)}=\frac{1}{\tau}(\frac1k\sigma_kT_{11}\rho_k+\frac1k\sigma_lT_{21}\rho_k
    +\frac1l\sigma_kT_{12}\rho_l+\frac1l\sigma_lT_{22}\rho_l),\\
    &s_4^{(-1,0)}=\frac{1}{\tau}(\sigma_kQ_{11}\rho_k^*+\sigma_lQ_{21}\rho_k^*
    +\sigma_kQ_{12}\rho_l^*+\sigma_lQ_{22}\rho_l^*).
\end{align}
\esb

\subsection{Jordan block solution}

When $\bs K_1$ is a   $2\times 2$ Jordan matrix, we have
\begin{align*}
    \bs K_1=\begin{pmatrix}
        k & 0\\
        1 & k
    \end{pmatrix},~~
    \bs r_1=\begin{pmatrix}
        \rho  \\
        \pa_k\rho
    \end{pmatrix},~~
    \bs s_1=\begin{pmatrix}
        \sigma  \\
        \pa_k\sigma
    \end{pmatrix},
\end{align*}
and
\begin{align*}
   \bs M_1=\begin{pmatrix}
        m_{11} & m_{12} \\
        m_{21} & m_{22}
    \end{pmatrix}
    = \begin{pmatrix}
        -\frac{\rho\sigma}{|k|^2+1}+\frac{k\rho\partial_k\sigma}{(|k|^2+1)^2} & -\frac{\rho\partial_k\sigma}{|k|^2+1} \\
        -\frac{\pa_k\rho\sigma}{|k|^2+1}+\frac{k^*\rho\sigma+\rho\pa_k\sigma}{(|k|^2+1)^2}
        -\frac{|k|^2\rho\pa_k\sigma}{(|k|^2+1)^3} & -\frac{\partial_k\rho\pa_k\sigma}{|k|^2+1}+\frac{k^*\rho\partial_k\sigma}{(|k|^2+1)^2}
    \end{pmatrix}, 
\end{align*}
where $\rho$ and $\sigma$ are defined as in \eqref{ps}.
In this case,
\[
 \bs T=(\bs I-\bs M_1^*\bs M_1)^{-1}
    =\frac{1}{\tau}\begin{pmatrix}
        T_{11} & T_{12} \\
        T_{21} & T_{22}
    \end{pmatrix}
=\frac{1}{\tau}
    \begin{pmatrix}
        1-m_{21}^*m_{12}-|m_{22}|^2 & m_{11}^*m_{12}+m_{12}^*m_{22} \\
        m_{21}^*m_{11}+m_{22}^*m_{21} & 1-|m_{11}|^2-m_{12}^*m_{21}
    \end{pmatrix},\]
where
\begin{align*}
    \tau&=|\bs I-\bs M_1^*\bs M_1|=1-m_{21}^*m_{12}-m_{12}^*m_{21}-|m_{11}|^2-|m_{22}|^2 \\
    &+|m_{11}|^2|m_{22}|^2+|m_{12}|^2|m_{21}|^2
    -m_{11}m_{22}m_{12}^*m_{21}^*-m_{11}^*m_{22}^*m_{21}m_{12} \\
    &=1-2\mr{Re}[m_{12}^*m_{21}+m_{11}m_{22}m_{12}^*m_{21}^*]-|m_{11}|^2-|m_{22}|^2
    +|m_{11}|^2|m_{22}|^2+|m_{12}|^2|m_{21}|^2.
\end{align*}
Then $\bs v$ can be given via formula \eqref{s134} with \eqref{PQs},
while all $m_{ij}$'s and $T_{ij}$'s should be taken from this subsection.

\end{appendices}

\end{document}